\documentclass{jac}
\usepackage{amsmath}
\usepackage{amssymb}
\usepackage{array}
\usepackage{enumerate}
\usepackage{tabularx}
\usepackage{xcolor}
\usepackage{xspace}
\usepackage{graphicx}

\bgroup
   \divide\count0 by 60
  \multiply\count0 by 60
  \advance\count1 by -\count0
\xdef\daytime{\the\count2:\ifnum\count1<10 0\fi \the\count1}
\egroup

\def\msbl{\ensuremath{\langle}}
\def\msbr{\ensuremath{\rangle}}
\def\lsbl{\ensuremath{\vert}}
\def\lsbr{\ensuremath{\vert}}

\def\reg#1{\emph{#1}}
\def\Rcomp{\reg{comp}\xspace}
\def\Rleft{\reg{left}\xspace}

\def\Rright{\reg{right}\xspace}
\def\RRIGHT{\reg{right2}\xspace}
\def\Rmax{\reg{max}\xspace}
\def\Rnext{\reg{next}\xspace}
\def\spacer{\vrule depth 0.6ex height 3ex width 0ex}
\def\spacerb{\vrule depth 1.2ex height 2.4ex width 0ex}
\def\L#1{\hbox to 2.5em{\spacer\hss\ensuremath{\overleftarrow{#1}}\hss}}
\def\R#1{\hbox to 2.5em{\spacer\hss\ensuremath{\overrightarrow{#1}}\hss}}
\def\N#1{\hbox to 2.5em{\spacer\hss\ensuremath{{#1}}\hss}}

\begin{document}
\title{Real-time Sorting of Binary Numbers on One-dimensional CA}

\author[lab1]{Th. Worsch}{Thomas Worsch}
\address[lab1]{KIT (Karlsruhe Institute of Technology)}
\email{worsch@kit.edu}
\author[lab2]{H. Nishio}{Hidenosuke Nishio}
\address[lab2]{Kyoto University}
\email{yra05762@nifty.com}	



\begin{abstract}\noindent
  A new fast (real time) sorter of binary numbers by one-dimensional
  cellular automata is proposed.  It sorts a list of $n$ numbers
  represented by $k$-bits each in exactly $nk$ steps. This is only one
  step more than a lower bound.
\end{abstract}
\maketitle

\section{Introduction}
\label{sec:intro}

Sorting is one of the most fundamental subjects of computer science
and many sorting algorithms including sorting arrays and networks can
for example be found in volume~3 of Knuth's TAOCP
\cite{Knuth98}. However, for cellular automata there are only a few
papers on this important topic. It should be pointed out that the
algorithm described by one of the authors in an earlier paper
\cite{Nishio75} has running $3nk$ (despite the title starting with the
words ``real time''). 

We are not aware of any speedup techniques which would allow to turn
this CA or any other solving the problem into one running in
real-time, i.\,e.\ exactly the number of steps which is the length of
the input.  In the present paper we propose a sorting algorithm of
binary numbers and its implementation on one-dimensional CA with
nearest neighbors, which sorts $n$ numbers of $k$ bits each in exactly
$nk$ steps.

Sequential comparison based sorting algorithms need time $\Omega(n\log
n)$ where $n$ is the number of elements to be sorted and it is assumed
that each comparison can be done in constant time independent of the
size of the elements. The latter assumption is also usually made for
parallel sorting algorithms. On linear arrays odd-even transposition
sort needs exactly $n$ steps. But there the additional assumption is
made that from the beginning each processor knows the parity of its
own address (assuming those are e.\,g.\ $1$ to $n$). Of course on
parallel models from the second machine class \cite{EmdeBoas90} like
PRAM there are algorithms running in poly-logarithmic (or even
logarithmic) time. But for these models one has to assume non-local
communication in constant time when embedded into Euclidean space
\cite{Schorr83}.

The rest of the paper is organized as follows.  In
Section~\ref{sec:statements} we precisely state the problem and the
results obtained in this paper. In Section~\ref{sec:lower-bound} we
give a short proof for the lower bound of the sorting problem.  The
main aspects of our algorithm are presented in
Section~\ref{sec:base-alg}. The first version does not achieve a
running time which matches the lower bound. For that two modifications
are needed which are described in Section~\ref{sec:fast-alg}.

\section{Statement of problem and results}
\label{sec:statements}

We are considering one-dimensional CA with von Neumann neighborhood of
radius $1$ and assume that the reader is familiar with these
concepts. Since we will not define our CA on such a low level there is
no need to introduce any related formalism.

We also assume that the reader is familiar with the firing squad
synchronization problem \cite{Umeo05}. If a block of $k$ cells needs
to be synchronized and there are generals at both ends, then
synchronization can be achieved in exactly $k$ steps.

The inputs for our CA are provided as finite words with all
surrounding cells in a quiescent state. Those cells will never be
used during computations.

The inputs which have to be processed by our CA are $n$ numbers of
equal length $k$, with the most and least significant bits marked as
such.

\begin{problem}
  \label{prob:sorting}
  The input alphabet is $A=\{ 0, 1, \msbl 0, \msbl 1, 0\lsbr, 1\lsbr \}$.

  Each input $w$ that has to be processed properly is of the form 
  \[
  w=w_1\cdots w_n=\msbl x_1 y_1 z_1\lsbr\cdots \msbl x_n y_n z_n \lsbr
  \]
  for some $n\geq 1$ and $k\geq 2$ where all $x_i, z_i\in\{0,1\}$ and
  all $y_i\in\{0,1\}^{k-2}$. Each $w_i$ is the binary representation
  of a non-negative integer (also denoted $w_i$) with most significant
  bit $x_i$ and least significant bit $z_i$.

  For every such input after a finite number of steps a stable
  configuration of the form $w_{\sigma(1)} \cdots w_{\sigma(n)}$ has to be
  reached where $\sigma$ is a permutation of the numbers $1,\dots,n$ such
  that $w_{\sigma(i)}\leq w_{\sigma(i+1)}$ holds for all $1\leq i<n$.
\end{problem}
We note that it would have been sufficient to mark either the most or
the least significant bits, because the other end of each number can
then always be identified by looking at neighbor cells.

The above problem statement also excludes the case of $1$-bit
inputs. For those the ``traffic rule'' 184 can be easily extended to
do sorting, taking into account quiescent neighbors. The resulting CA
works as follows: A cell in state $1$ ($0$) becomes $0$ ($1$) if its
right (left) neighbor is $0$, otherwise it keeps its current state.

In the following we always call a sequence of cells which initially
stores one input number $w_i$ the \emph{block} $i$ (or simply a
block).

It is clear that it can be necessary to move a number from block $1$
to block $n$. This immediately gives a lower bound of $(n-1)k$ steps
for the sorting time. We shall see, that one can do slightly better:

\begin{theorem}
  \label{thm:lower-bound}
  Every CA solving Problem~\ref{prob:sorting} needs at least time
  $nk-1$.
\end{theorem}

Until now the fastest sorting algorithms known needed time $cnk$ for
some constant $c>1$ with no obvious possibility to speed up the
computation to run in $nk$ steps. The main contribution of the present
paper therefore is the following:

\begin{theorem}
  \label{thm:upper-bound-nk}
  There is a CA (which does not depend on $n$ or $k$) with von Neumann
  neighborhood of radius $1$ solving the Sorting
  Problem~\ref{prob:sorting} in exactly $nk$ steps.
\end{theorem}

\section{Lower bound on sorting time}
\label{sec:lower-bound}

First consider the input $w=w_1w_2\cdots w_n$ where $w_1=\msbl 0
1^{k-2}1\lsbr$ and $w_2=\cdots=w_n=\msbl{1}0^{k-2}{0}\lsbr$. Clearly this
input sequence is already sorted and the rightmost bit of the output
is a $0$.

If on the other hand we flip the leftmost bit of the first block only
and consider the input
\[
w'_1\cdot w_2\cdots w_n=\msbl 1 1^{k-2}1\lsbr
\cdot\msbl{1}0^{k-2}{0}\lsbr\cdots\msbl{1}0^{k-2}{0}\lsbr
\]
then the correct sorted output is
\[
w_2\cdots w_n \cdot w'_1 =
\msbl{1}0^{k-2}{0}\lsbr\cdots\msbl{1}0^{k-2}{0}\lsbr \cdot \msbl 1
1^{k-2}1\lsbr
\]
That is, by changing only the leftmost bit of the input the rightmost
bit of the output must change. Hence no CA correctly solving the
sorting problem can be faster then the distance between leftmost and
rightmost bit which is $nk-1$.

This proves Theorem~\ref{thm:lower-bound}.

\section{The base sorting algorithm}
\label{sec:base-alg}

The goal of this section is to prove a weakened version of
Theorem~\ref{thm:upper-bound-nk}.

\begin{lemma}
  \label{lem:upper-bound-nk+k}
  There is a CA (which does not depend on $n$ or $k$) with the von
  Neumann neighborhood of radius $1$ solving the sorting
  Problem~\ref{prob:sorting} in exactly $k+nk$ steps.
\end{lemma}

Before going into details and explaining some aspects of the CA on the
cell level, we describe the main idea on the level of numbers. In
particular we will employ a well-known simple algorithm for parallel
sorting.

\subsection{Odd-even transposition sort}
\label{subsec:odd-even}

Throughout this section, one can assume that $N$ is an even number;
this is the case needed for the CA below.

Assume that $N$ numbers $a_1, a_2, \dots, a_N$ are given, arranged in
an array of processors. In addition each processor (or each number)
has a direction (indicated by arrows below).  In each step each pair
of adjacent processors whose arrows point to each other exchange their
numbers. Both processors compare the two numbers; the left one keeps
the smaller number, the right one the larger number, and both change
their direction to the other neighbor. Thus one step can for example
look like this:

\begin{center}
   \begin{tabular}{l@{ }|*{9}{>{$}c<{$}|}}
     \cline{2-7}\cline{9-10}
     $t$ & \L{a_1} & \R{a_2} & \L{a_3} & \R{a_4} & \L{a_5} & \R{a_6} & \dots\dots & \L{a_{N-1}} & \R{a_N}\\
     \cline{2-7}\cline{9-10}
     \multicolumn{1}{c}{ }\\[-1.5ex] 
    \cline{2-7}\cline{9-10}
     $t+1$&\R{b_1} & \L{b_2} & \R{b_3} & \L{b_4} & \R{b_5} & \L{b_6} & \dots\dots & \R{b_{N-1}} & \L{b_N}\\
     \cline{2-7}\cline{9-10}
     \multicolumn{1}{c}{ }\\
  \end{tabular}

  where $b_i=
  \begin{cases}
    \min(a_i,a_{i+1}) & \text{ if } a_i \text{ points to the right}\\
    \max(a_{i-1},a_i) & \text{ if } a_i \text{ points to the left} 
  \end{cases}
  $. \\[1ex]
\end{center}

\noindent
A missing neighboring number to the left is treated as if it were
$-\infty$ and a missing neighboring number to the right is treated as
if it were $\infty$.

It is known that odd-even transposition sort always produces the
correct result after exactly $N$ steps (see e.g.\cite{Knuth98}).

\subsection{Outline of the base algorithm}
\label{subsec:base-outline}

The algorithm which will be the basis for the improved construction in
Section~\ref{sec:fast-alg} will simply work as follows. Given an input
of the form
\begin{center}
  \begin{tabular}{l@{ }|*{9}{>{$}c<{$}|}}
    \cline{2-5}\cline{7-8}
     & \N{w_1} & \N{w_2} & \N{w_3} & \N{w_4} & \dots\dots & \N{w_{n-1}} & \N{w_n}\\
    \cline{2-5}\cline{7-8}
  \end{tabular}
\end{center}
first each number is copied and the two copies get opposite directions
assigned. This will take $k$ steps in the CA.  Instead of storing two
copies side by side they are stored in parallel:
\begin{center}
  \begin{tabular}{l@{ }|*{9}{>{$}c<{$}|}}
    \cline{2-5}\cline{7-8}
     & \L{w_1} & \L{w_2} & \L{w_3} & \L{w_4} & \dots\dots & \L{w_{n-1}} & \L{w_n}\\
    \cline{2-5}\cline{7-8}
     & \R{w_1} & \R{w_2} & \R{w_3} & \R{w_4} & \dots\dots & \R{w_{n-1}} & \R{w_n}\\
    \cline{2-5}\cline{7-8}
  \end{tabular}
\end{center}
These are now treated as $N=2n$ numbers that are sorted using odd-even
transposition sort. In the end one would get
\begin{center}
  \begin{tabular}{l@{ }|*{9}{>{$}c<{$}|}}
    \cline{2-5}\cline{7-8}
     & \L{w_{\sigma(1)}} & \L{w_{\sigma(2)}} & \L{w_{\sigma(3)}} & \L{w_{\sigma(4)}} 
     & \dots\dots & \L{w_{\sigma(n-1)}} & \L{w_{\sigma(n)}}\\
    \cline{2-5}\cline{7-8}
     & \R{w_{\sigma(1)}} & \R{w_{\sigma(2)}} & \R{w_{\sigma(3)}} & \R{w_{\sigma(4)}} 
     & \dots\dots & \R{w_{\sigma(n-1)}} & \R{w_{\sigma(n)}}\\
    \cline{2-5}\cline{7-8}
  \end{tabular}
\end{center}
where $\sigma$ again denotes the ``sorting permutation'' as in
Problem~\ref{prob:sorting}.

During the last sorting step the lower parts and the arrows are
deleted, and the required output is obtained.


It will become clear in the next subsection why it is actually useful
to first copy each number and then seemingly spend twice as much time
for the $N=2n$ sorting steps.  It will be shown that each such step can
be implemented in the CA in $k/2$ steps. Hence the total running time
of the CA for this base version of the algorithm will be $k+nk$ steps.

\subsection{Outline of the CA for the base algorithm}
\label{subsec:base-ca}

In this section we will describe how the base sorting algorithm can be
implemented on a CA.  It will need $n+1$ phases each of which needs
exactly $k$ steps. First comes a setup phase followed by phases $1$,
\dots, $n$. In order to avoid more complicated descriptions,
throughout the rest of the paper we assume that $k$ is even. 

If $k$ is odd the middle cell of a block plays the role of the two
middle cells one would have for the (even) case $k+1$. In that case
synchronization of a block using generals at both ends needs at least
time $2\frac{k+1}{2}-2 = k-1$ and hence is possible in time $k$ (as is
the case for even $k$).

\begin{algorithm}[Setup phase]
  \label{alg:base-setup}
  \leavevmode\mbox{ }
  \begin{enumerate}
  \item During the setup phase the following tasks are carried
    out in each block:
    \begin{itemize}
    \item By sending a signal from the left and the right end of the
      block the two middle cells are found. In each resulting
      \emph{sub-block} the leftmost and the rightmost cell are marked
      as such. They are called $L$ and $R$ respectively.
    \item Using an additional register the mirror image of the input
      number is computed. Below we call the
      register holding the original value \Rleft and the registers
      with the mirrored value \Rright.

      The numbers in the \reg{left} registers will play the role of
      the $\overleftarrow{w_i}$ and the numbers in the \reg{right}
      registers the role of the $\overrightarrow{w_i}$ used in the
      previous subsection.
    \item Using synchronization, the preliminary phase is stopped
      after $k$ steps.
    \end{itemize}
  \item The leftmost and the rightmost cell of the whole input are set
    up as generals. Starting with the first step of phase $1$ an
    algorithm is started to synchronize all $nk$ cells after $nk$
    steps.
  \end{enumerate}
\end{algorithm}

\noindent
A concrete example is shown in Figure~\ref{fig:ex-prep} for two
$6$-bit numbers. The borders between sub-blocks are shown as double
vertical lines. It can be noted that we also use markers for the most
and least significant bits of the mirrored numbers.

\begin{figure}[ht]
  \centering
  \begin{tabular}{l@{\quad}|*{4}{|*{3}{>{$}c<{$}|}}|}
    \multicolumn{1}{c}{ }& \multicolumn{1}{c}{$L$} &    \multicolumn{1}{c}{ }& \multicolumn{1}{c}{$R$} &
    \multicolumn{1}{c}{$L$} &    \multicolumn{1}{c}{ }& \multicolumn{1}{c}{$R$} &
    \multicolumn{1}{c}{$L$} &    \multicolumn{1}{c}{ }& \multicolumn{1}{c}{$R$} &
    \multicolumn{1}{c}{$L$} &    \multicolumn{1}{c}{ }& \multicolumn{1}{c}{$R$} \\
    \cline{2-13}
    \spacerb\Rleft & \msbl 1 & 0 & 0 & 0 & 0 & 0\lsbr & \msbl 0 & 1 & 1 & 1 & 1 & 1\lsbr \\
    \cline{2-13}
    \spacerb\Rright & \lsbl 0 & 0 & 0 & 0 & 0 & 1\msbr & \lsbl 1 & 1 & 1 & 1 & 1 & 0\msbr \\
    \cline{2-13}
  \end{tabular}
  \caption{An example configuration for two $6$-bit numbers after the
    setup phase.}
\label{fig:ex-prep}
\end{figure}

In addition to the registers \Rleft and \Rright the $L$- and $R$-cells
of each sub-block will make use of a register \Rcomp which will hold a
(preliminary) comparison result (see also Figure~\ref{fig:comp}).
Each \Rcomp register can hold one of the values $=$, $<$ or $>$.
Their use is described in Algorithm~\ref{alg:base-c2} below.

The core idea is the following: 
\begin{enumerate}[C1.]
\item Bits are shifted in the \Rleft and \Rright registers in the
  corresponding directions.
\item Whenever the most significant bits of two numbers arrive in a
  pair of adjacent $R\Vert L$-cells, the numbers are compared
  sequentially bit by bit. The smaller number will be directed to move
  to the left and the larger one will be directed to move to the
  right.
\end{enumerate}

\noindent
Part C1 basically means that numbers are unconditionally shifted
everywhere except at $R\Vert L$ pairs. Since those are located at
distance $k/2$ and numbers have length $k$ this might look suspicious
at first sight, because in general a number simultaneously gets
compared at two such pairs. It will become clear later why this does
not pose any problems. Ignoring it for the moment, C1 is easy to
implement:

\begin{algorithm}[Implementation of C1]
  \label{alg:base-c1}
  \leavevmode\mbox{ }
  \begin{enumerate}
  \item The cells which are \emph{not} an $R$- or $L$-cell have a very
    simple behavior.
    \begin{itemize}
    \item The \Rleft register gets its content from the \Rleft
      register of the right neighbor.
    \item The \Rright register gets its content from the \Rright
      register of the left neighbor.
    \end{itemize}
  \item Analogously the \Rleft register of an $L$-cell gets its
    content from the \Rleft register of the right neighbor and the
    \Rright register of an $R$-cell gets its content from the \Rright
    register of the left neighbor.
  \item The same holds for the \Rleft register of a $R$-cell and the
    \Rright register of a $L$-cell if the \Rcomp register of that cell
    has value $=$.

    First of all, this part is needed during the first sub-phase of
    phase $1$ when the $R\Vert L$ pairs in the middle of a block still
    have no meaning. As will be seen this requirement is also
    consistent with the rules for later sub-phases.
  \end{enumerate}
\end{algorithm}

\noindent
Each of the phases $1$, \dots, $n$ is subdivided into two sub-phases of
$k/2$ steps each.

We will now describe how the comparison of numbers is done.  For this
we use $R.\Rleft$ to denote the \Rleft register of the $R$-cell and
similarly for the other cases.  It will be seen that all information
needed to update $R.\Rcomp$ are also available in the neighboring
$L$-cell, so that the invariant $R.\Rcomp=L.\Rcomp$ can be
maintained. Hence it suffices to describe the case of $R.\Rcomp$:

\begin{algorithm}[Implementation of C2]
  \label{alg:base-c2}
  If the two bits in $R.\Rright$ and $L.\Rleft$ are most significant
  ones, then the new value of $R.\Rcomp$ is determined as follows:
  \begin{equation}
    R.\Rcomp \text{ becomes }
    \begin{cases}
      < & \text{ if } R.\Rright< L.\Rleft\\
      = & \text{ if } R.\Rright= L.\Rleft\\
      > & \text{ if } R.\Rright> L.\Rleft\\
    \end{cases} \label{rule:comp}
  \end{equation}
  If the two bits to be compared are not most significant ones, the
  new value of $R.\Rcomp$ is determined as follows:
  \begin{itemize}
  \item If $R.\Rcomp$ already has value $<$ or $>$ it is not changed.
  \item If $R.\Rcomp$ has old value $=$ its new value is determined
    according to rule~\ref{rule:comp} above.
  \end{itemize}
  It remains to define how the new value for $R.\Rleft$ is
  computed. That is most easily described as depending on the just
  defined \emph{new} value of $R.\Rcomp$.:
  \[
  R.\Rleft \text{ becomes }
  \begin{cases}
    R.\Rright & \text{ if } R.\Rcomp \text{ is now } <\\
    R.\Rright & \text{ if } R.\Rcomp \text{ is now } = \\
    L.\Rleft & \text{ if } R.\Rcomp \text{ is now } >\\
  \end{cases}
  \]
  Dually the new value for $L.\Rright$ depends on the \emph{new} value
  of $L.\Rcomp$ (remember that always $R.\Rcomp=L.\Rcomp$):
  \[
  L.\Rright \text{ becomes }
  \begin{cases}
    L.\Rleft & \text{ if } L.\Rcomp \text{ is now } <\\
    L.\Rleft & \text{ if } L.\Rcomp \text{ is now } = \\
    R.\Rright & \text{ if } L.\Rcomp \text{ is now } >\\
  \end{cases}
  \]
\end{algorithm}

\noindent 
Figure~\ref{fig:comp} shows the relevant parts of computations for the
comparison of two numbers. In the left part of the figure initially
the larger number is on the left, in the right part the smaller number
is on the left. After the comparison in both cases the smaller number
is on the left. We remind the reader that at the left resp.\ right end
of the complete input a missing number is treated as $-\infty$ resp.\
$\infty$. Hence a number arriving at a border is simply reflected.

\begin{figure}[ht]
  \centering
  \begin{tabular}{l@{\quad}*{2}{|*{4}{>{$}c<{$}|}}c|*{4}{>{$}c<{$}|}|*{4}{>{$}c<{$}|}}
    \multicolumn{1}{c}{ }
    & \multicolumn{1}{c}{ } & \multicolumn{1}{c}{ } & \multicolumn{1}{c}{ } & \multicolumn{1}{c}{$R$} 
    & \multicolumn{1}{c}{$L$} & \multicolumn{1}{c}{ } & \multicolumn{1}{c}{ } & \multicolumn{1}{c}{ } 
    & \multicolumn{1}{c}{\qquad} 
    & \multicolumn{1}{c}{ } & \multicolumn{1}{c}{ } & \multicolumn{1}{c}{ } & \multicolumn{1}{c}{$R$} 
    & \multicolumn{1}{c}{$L$} & \multicolumn{1}{c}{ } & \multicolumn{1}{c}{ } & \multicolumn{1}{c}{ } \\
    \cline{2-9} \cline{11-18}
    \spacerb\Rleft &  &  &  & & \msbl 0 & 0 & 1 & 1 & & &  &  & & \msbl 0 & 1 & 0 & 1  \\
    \cline{2-9} \cline{11-18}
    \spacerb\Rright &  1 & 0 & 1 & 0\msbr & & &  & & &  1 & 1 & 0 & 0\msbr & & &  & \\
    \cline{2-9} \cline{11-18}
    \spacerb\Rcomp &  & &  & & & & & &  & &  & & & & & & \\
    \cline{2-9} \cline{11-18}
    \multicolumn{1}{c}{ }\\[-2ex]
    \cline{2-9} \cline{11-18}
    \spacerb\Rleft &  &  &  & \msbl 0 & 0 & 1 & 1 & & & & & & \msbl 0 & 1 & 0 & 1 & \\
    \cline{2-9} \cline{11-18}
    \spacerb\Rright & &  1 & 0 & 1 & 0\msbr & & & & & &  1 & 1 & 0 & 0\msbr & & & \\
    \cline{2-9} \cline{11-18}
    \spacerb\Rcomp &  & &  & =&= & & & & &  & &  & =&= & & & \\
    \cline{2-9} \cline{11-18}
    \multicolumn{1}{c}{ }\\[-2ex]
    \cline{2-9} \cline{11-18}
    \spacerb\Rleft &  &  &  \msbl 0 & 0 & 1 & 1 & & &  & &&\msbl 0 &0 &0 &1 & & \\
    \cline{2-9} \cline{11-18}
    \spacerb\Rright & & &  1 & 0 & 1 & 0\msbr & &  & & && 1 &1 &1 &0\msbr & & \\
    \cline{2-9} \cline{11-18}
    \spacerb\Rcomp &  & &  & >&> & & & &  & &  & & <&< & & & \\
    \cline{2-9} \cline{11-18}
    \multicolumn{1}{c}{ }\\[-2ex]
    \cline{2-9} \cline{11-18}
    \spacerb\Rleft &  &  \msbl 0 & 0 & 1 & 1 & & & & & &\msbl0 &0 &1 &1 & & &  \\
    \cline{2-9} \cline{11-18}
    \spacerb\Rright & & & &  1 & 0 & 1 & 0\msbr &  & & & & &1 &0 &1 &0\msbr & \\
    \cline{2-9} \cline{11-18}
    \spacerb\Rcomp &  & &  & >&> & & & &  & &  & & <&< & & & \\
    \cline{2-9} \cline{11-18}
    \multicolumn{1}{c}{ }\\[-2ex]
    \cline{2-9} \cline{11-18}
    \spacerb\Rleft &  \msbl 0 & 0 & 1 & 1 & & & &  & &\msbl0 &0 &1 &1 & & & & \\
    \cline{2-9} \cline{11-18}
    \spacerb\Rright & & & & & 1 & 0 & 1 & 0\msbr  & & & & & &1 &0 &1 &0\msbr \\
    \cline{2-9} \cline{11-18}
    \spacerb\Rcomp &  & &  & >&> & & & &  & &  & & <&< & & & \\
    \cline{2-9} \cline{11-18}
  \end{tabular}
  \caption{Two comparisons of $4$ bits. On the left hand side the
    number coming from the left (in the \Rright\ registers) is larger,
    on the right hand side the number coming from the right (in the
    \Rleft\ registers). If the complete numbers were longer, the
    comparisons would continue analogously. }
  \label{fig:comp}
\end{figure}

The following picture may be helpful: When the smaller number comes
from the left and the larger from the right, then from the first (most
significant) bit both numbers are reflected at the border between the
$R$- and the $L$-cell. When the larger number comes from the left and
the smaller from the right, then from the first (most significant) bit
both numbers pass through the border. This is a correct picture even
if the numbers have identical higher order bits, and hence in the
beginning no cell knows which is the present case. Readers are
encouraged to check Figure~\ref{fig:comp} again.

At least from a formal point of view it is now straightforward to put
the pieces together. An even better intuition of how the algorithm
works may arise in the subsequent Subsection~\ref{subsec:base-correct}
when the correctness of the algorithm will be shown.

\begin{samepage}
\begin{algorithm}
  \label{alg:base}
  \leavevmode\mbox{ }
  \begin{itemize}
  \item First the setup phase is done as described in
    Algorithm~\ref{alg:base-setup}. This phase takes $k$ steps. It is
    stopped at the correct time using a synchronization algorithm in
    each block separately. Both ends of each block have to act as
    generals in order to achieve the required synchronization time.
  \item Once all cells are synchronized they will work as described in
    Algorithms~\ref{alg:base-c1} and~\ref{alg:base-c2}: All numbers
    are shifted to the left or right, and whenever two most
    significant bits meet, the sequential comparison of the two
    numbers is started. The smaller number is sent to the left and the
    larger to the right.
  \item This is repeated until the synchronization started immediately
    after the setup phase fires all $nk$ cells after $nk$ steps.

    It will be shown that at that point in time the \Rleft registers
    contain the sorted numbers.
  \end{itemize}
\end{algorithm}
\end{samepage}

\subsection{Correctness of the base sorting algorithm}
\label{subsec:base-correct}

The correctness of algorithm~\ref{alg:base} is essentially due to the
correctness of odd-even transposition sort.  This is basically a proof
by induction. The main parts are stated in the following Lemma.

\begin{lemma}
  \label{lem:base-correct-core}
  Let $\mathcal{X}=\lsbl xX\msbr$ and $\mathcal{Y}=\msbl Yy\lsbr$ be
  two numbers with the $k/2$ higher order bits denoted by capital
  letters and the $k/2$ lower order bits denoted by small
  letters. Similarly let $\mathcal{A}=\msbl Aa\lsbr$ be the minimum of
  $\mathcal{X}$ and $\mathcal{Y}$ and $\mathcal{B}=\lsbl bB\msbr$ be
  the maximum. In other words one basic compare-and-exchange step of
  odd-even transposition sort transforms the pair $\mathcal{X},
  \mathcal{Y}$ into the pair $\mathcal{A}, \mathcal{B}$.
  
  If the most significant bits of $X\msbr$ and $\msbl Y$ meet in an
  $R\Vert L$-pair (with the other higher order bits following) and if
  after $k/2$ steps the lower order bits $\lsbl x$ and $y\lsbr$ arrive
  in the correct order, then during the first $k/2$ steps the higher
  order bits of $\msbl A$ and $B\msbr$ will be produced moving to the
  right directions, followed by the lower order bits $a\lsbr$ and
  $\lsbl b$ afterwards.
\end{lemma}

This is basically a restatement of the construction from
Algorithm~\ref{alg:base-c2}.

Since the configuration produced by the setup phase corresponds
to the initial configuration for the odd-even transposition sort, and
since the preconditions of the if-statement in
Lemma~\ref{lem:base-correct-core} are met, an induction teaches that
in particular the higher order bits of each number after $t$ phases
are in the sub-block corresponding to its position in odd-even
transposition sort after $t$ sorting steps.

\begin{corollary}
  \label{cor:base-correct}
  For all input sequences $w_1,\dots, w_n$ Algorithm~\ref{alg:base}
  does sort the numbers as required in Problem~\ref{prob:sorting}.
\end{corollary}

\begin{proof}
  Since odd-even transposition sort does sort $N$ numbers in $N$
  sorting steps, it immediately follows from
  Lemma~\ref{lem:base-correct-core} that at the end of phase $n$ of
  Algorithm~\ref{alg:base} the higher order bits of each of the $n$
  input numbers and of its $n$ copies are in the correct blocks. That
  is, there are the same higher order bits of a number $w_i$ in each
  block twice, once in the \Rleft registers and once in the \Rright
  registers.

  Furthermore it is clear that the most significant bit of the number
  stored in the \Rleft (resp.\ \Rright) registers is in the $L$-cell
  (resp.\ $R$-cell) of the \emph{block} (not sub-block). This implies
  that the lower order bits are in the same block:

  \begin{tabular}{l@{\quad}cc}
    & \multicolumn{1}{l}{$L$} & \multicolumn{1}{r}{$R$} \\
    \cline{2-3}
    \spacerb\Rleft & \multicolumn{1}{||c||}{$\msbl k/2$ higher order bits of $w_i$} 
    & \multicolumn{1}{c||}{$k/2$ lower order bits of $w_i\lsbr$}  \\
    \cline{2-3}
    \spacerb\Rright & \multicolumn{1}{||c||}{$\lsbl k/2$ lower order bits of $w_i$}
    & \multicolumn{1}{c||}{$k/2$ higher order bits of $w_i\msbr$}  \\
    \cline{2-3}
    \\
  \end{tabular}

  \noindent
  Therefore the \Rleft registers of the full block hold the correct
  value.
\end{proof}

\section{A sorting algorithm matching the lower bound}
\label{sec:fast-alg}

We will save $k$ steps of the running time of Algorithm~\ref{alg:base}
by starting with some comparisons not after $k$ but already after
$k/2$ steps and stopping the odd-even transposition sort $k/2$ steps
earlier. A detailed description will be given in
Section~\ref{subsec:speedup}. The resulting algorithm still computes
the correct output \emph{except} for the rightmost block. This will be
fixed in Section~\ref{subsec:fix-right}.

\subsection{Speeding up the algorithm}
\label{subsec:speedup}

We describe the fast algorithm as three changes to Algorithm~\ref{alg:base}.

The first change is simple: Since we want to have the result after
$nk$ steps instead of $k+nk$, the synchronization of all $nk$ cells is
not started after the setup phase, but in the very first step.

The second change concerns the computation of the mirror of a bit
string as required by Algorithm~\ref{alg:base-setup}. It can be
implemented by shifting the original to the left, and letting the
$L$-cell of the block act as reflector sending bits back to the right
in the \Rright cells. This means that after $k/2$ steps the lower
order bits of an input number have arrived in the \Rleft registers of
the left sub-block and the higher order bits are its \Rright
registers. It is useful if the shift to the left is not done using a
temporary register but \Rleft. In Figure~\ref{fig:mirror} the
resulting process is shown for two adjacent $6$-bit numbers. It can be
seen that due to the simultaneous shift to the left, already after
$k/2$ steps for the first time most significant bits meet. (This also
determines the border of the sub-blocks.) Thus comparisons can be
started $k/2$ steps earlier. The rightmost block now needs some
special attention. Analogously to the other blocks we assume that
symbols representing the ``number $\infty$'' are shifted to the left
from the rightmost cell. This is also depicted in
Figure~\ref{fig:mirror}.

This completes the second change to Algorithm~\ref{alg:base}.

\begin{figure}[ht]
  \centering
  \begin{tabular}{l@{\quad}||*{6}{>{$}c<{$}|}|*{6}{>{$}c<{$}|}|}
    \cline{2-13}
    \spacerb\Rleft & \msbl a_6 & a_5 & a_4 & a_3 & a_2 & a_1\lsbr& \msbl b_6 & b_5 & b_4 & b_3 & b_2 & b_1\lsbr \\
    \cline{2-13}
    \spacerb\Rright & & & & & & & & & & & & \\
    \cline{2-13}
    \multicolumn{7}{c}{ }\\
    \cline{2-13}
    \spacerb\Rleft & a_5 & a_4 & a_3 & a_2 & a_1\lsbr & \msbl b_6 & b_5 & b_4 & b_3 & b_2 & b_1\lsbr&\infty \\
    \cline{2-13}
    \spacerb\Rright & a_6\msbr& & & & & & b_6\msbr & & & & & \\
    \cline{2-13}
    \multicolumn{7}{c}{ }\\
    \cline{2-13}
    \spacerb\Rleft & a_4 & a_3 & a_2 & a_1\lsbr & \msbl b_6 & b_5 & b_4 & b_3 & b_2 & b_1\lsbr& \infty& \infty\\
    \cline{2-13}
    \spacerb\Rright & a_5 & a_6\msbr& & & & & b_5 & b_6\msbr& & & & \\
    \cline{2-13}
    \multicolumn{7}{c}{ }\\
    \cline{2-13}
    \spacerb\Rleft & a_3 & a_2 & a_1\lsbr & \msbl b_6 & b_5 & b_4 & b_3 & b_2 & b_1\lsbr& \infty& \infty& \infty\\
    \cline{2-13}
    \spacerb\Rright & a_4 & a_5 & a_6\msbr& & & &  b_4 & b_5 & b_6\msbr& & & \\
    \cline{2-13}
  \end{tabular}
  \caption{Computing the mirror of two $6$-bit numbers. Already after
    $3$ steps most significant bits meet at the border between
    sub-blocks.}
  \label{fig:mirror}
\end{figure}

The third change is the most complicated. For the shifts to the right
\emph{two} registers are used, \Rright and \RRIGHT. The additional
register \RRIGHT is empty almost everywhere. After each sub-phase
there are only two adjacent sub-blocks in which \RRIGHT stores a
number:
\begin{enumerate}
\item The initialization takes place in the leftmost block: During the
  first $k/2$ steps the reflected bits are shifted to the right in
  \Rright and \RRIGHT. And this is done \emph{only} during the first
  $k/2$ steps, but not afterwards.
\item Then it happens for the first time that three most significant
  bits meet, two coming from the left and one coming from the right.

  In such a case the comparisons are done as follows:
  \begin{itemize}
  \item The largest of the three numbers is shifted to the right in
    \RRIGHT.
  \item The other two numbers are compared and shifted as described in
    Algorithm~\ref{alg:base-c2}.
  \end{itemize}
\end{enumerate}
A consequence of these rules is, that after $k$ steps the \RRIGHT
registers are used in the sub-blocks of block $1$, after $2k$ steps in
the sub-blocks of block $2$, etc.\ and after $nk$ steps in the
sub-blocks of block $n$, and nowhere else.

Why do these three changes lead to a result after $nk$ steps where in
all blocks except the rightmost one the \Rleft registers contain the
correct numbers? That the rightmost block can still be wrong can be
seen in examples.

First of all, reflecting $w_1$ twice make sure that there are really
$2n$ numbers which are sorted, each $w_i$ twice. Therefore in the end
each block will again contain twice the same number.  This would in
general not be the case if $w_1$ would be reflected only once, and the
argument below would fail.

Since we are still using odd-even transposition sort, it is clear that
after $k/2+nk$ steps the correct results are obtained everywhere, but
with the most significant bits in the middle of the blocks. Thus the
result would look like

\begin{center}
  \begin{tabular}{l@{\quad}cc}
    & \multicolumn{1}{l}{$L$} & \multicolumn{1}{r}{$R$} \\
    \cline{2-3}
    \spacerb\Rleft & \multicolumn{1}{||c||}{$k/2$ lower order bits of $w_i\lsbr$} 
    & \multicolumn{1}{c||}{$\msbl k/2$ higher order bits of $w_i$}  \\
    \cline{2-3}
    \spacerb\Rright & \multicolumn{1}{||c||}{$k/2$ higher order bits of $w_i\msbr$}
    & \multicolumn{1}{c||}{$\lsbl k/2$ lower order bits of $w_i$}  \\
    \cline{2-3}
    \\
  \end{tabular}
\end{center}

How can such a left sub-block arise? Without loss of generality assume
that the input numbers are pairwise different. Then in the left
neighboring full block a smaller number (or $-\infty$) is
present. Hence the lower left half must have been reflected and going
back $k/2$ steps, the bits must have been in the upper
part. Analogously the upper right half must have been in the lower
part. Except in the rightmost block (where the \RRIGHT registers are
in use) there is no other possibility than that the lower order bits
are in the remaining registers:

\begin{center}
  \begin{tabular}{l@{\quad}cc}
    & \multicolumn{1}{l}{$L$} & \multicolumn{1}{r}{$R$} \\
    \cline{2-3}
    \spacerb\Rleft & \multicolumn{1}{||c||}{$\msbl k/2$ higher order bits of $w_i$}
    & \multicolumn{1}{c||}{$k/2$ lower order bits of $w_i\lsbr$}   \\
    \cline{2-3}
    \spacerb\Rright & \multicolumn{1}{||c||}{$\lsbl k/2$ lower order bits of $w_i$}
    & \multicolumn{1}{c||}{$k/2$ higher order bits of $w_i\msbr$}  \\
    \cline{2-3}
    \\
  \end{tabular}
\end{center}
Thus the \Rleft registers hold the desired result.

In the rightmost block it can happen that lower order bits are not
stored in the \Rleft registers of the right sub-block but in the
\RRIGHT registers of the left sub-block.

\subsection{Determining the rightmost output block}
\label{subsec:fix-right}

In order to produce the largest number in the rightmost output block
we use a separate algorithm which has to be run in parallel to the one
described above. It will have finished after $nk$ steps.  Remember
that the input is a word
\[
w=w_1\cdots w_n=\msbl x_1 y_1 z_1\lsbr\cdots \msbl x_n y_n z_n \lsbr
\]
and the task is to have the maximum of the $w_i$ be stored in the
rightmost block in the end. This can be achieved as follows.

\begin{algorithm}
\label{alg:fast-max}
  \leavevmode\mbox{ }
  \begin{enumerate}
  \item During the $k$ steps of the setup phase a signal is sent
    from the right end of the input until it reaches the most
    significant bit of $w_n$, marking all cells as belonging to the
    last block.

    When the last block is synchronized after $k$ steps, all cells
    have received the information and know that they have an
    additional task.
  \item From the very first step \emph{all} cells shift their input to
    the right using an additional register.
    
    Hence after $1\cdot k$ steps the number $w_{n-1}$ reaches the last
    block, after $2\cdot k$ steps number $w_{n-2}$ reaches the last
    block, etc.\ and after $(n-1)k$ steps number $w_1$ reaches the
    last block.
  \item The cells in the rightmost block use two additional registers
    for storing numbers; call them \Rmax and \Rnext. Register \Rmax is
    initialized in the very first step with $w_n$, register \Rnext is
    marked as not holding a value. Whenever the rightmost block ends a
    phase, in register \Rnext the number is stored that has arrived
    from the left in \Rright because of the shifting.
  \item If at the beginning of a phase register \Rnext has a valid
    number the rightmost block computes $\Rmax \leftarrow \max(\Rmax,
    \Rnext)$. For this a signal is sent from left to right, that is
    from the most significant bit to the least significant bit,
    comparing \Rnext and \Rmax. (While this comparison takes place the
    next number is already arriving in \Rright.)

    As long as the same bit value is found in both registers nothing
    is changed and the signal moves one cell to the right.

    As soon as at some position for the first time different bit
    values are found, the following happens:
    \begin{itemize}
    \item If \Rnext has a $1$ bit, but \Rmax has a $0$ bit,
      \Rmax is smaller than \Rnext and this and all remaining
      bits are copied from \Rnext to \Rmax.
    \item If \Rnext has a $0$ bit, but \Rmax has a $1$ bit,
      \Rmax is larger than \Rnext and the signal is simply
      killed leaving \Rmax unchanged.
    \end{itemize}

  \end{enumerate}
\end{algorithm}

\noindent
It is straightforward to verify by induction that for $1\leq i<n$ after
phase $i$ one has
\[
  \Rmax = \max \{w_{n-j}\mid 0\leq j<i\} 
\]
and hence in the end $\Rmax = \max \{w_i \mid 1\leq i \leq n\}$ as
required.  Since $w_1$ is copied to \Rnext after $(n-1)k$ steps the
final correct value is stored in \Rmax $k$ steps later, i.e.\ after
$nk$ steps as required.

Taking together the changes to the base Algorithm~\ref{alg:base}
described in Section~\ref{subsec:speedup} and the additional algorithm
just described one gets a proof of Theorem~\ref{thm:upper-bound-nk}.

\section{Conclusion}
\label{sec:conclusion}

We have shown the sorting of $n$ numbers with $k$ bits can be achieved
in (almost) real-time. Thus the situation is very similar to the
firing squad synchronization problem: There is an algorithm which has
--- in our case except for one step --- a running time matching a lower
bound.

Clearly, the number of states per cell required by our algorithm is
finite but large, at least when compared to algorithms e.\,g.\ for the
synchronization problem. We do not know how much the set of states can
be reduced.

The authors gratefully acknowledge a number of suggestions by the
referees for improving the presentation of the algorithms.

\end{document}